\documentclass[letter,11pt]{article}

\usepackage{fullpage}
\usepackage[capitalise]{cleveref}
\usepackage[numbers]{natbib}
\usepackage{xspace}
\usepackage{xcolor}
\usepackage{graphicx}
\usepackage{amsmath,amssymb,amsthm}
\usepackage{tikz}
\usetikzlibrary{chains,arrows}
\usepackage{todonotes}

\DeclareSymbolFont{bbold}{U}{bbold}{m}{n}
\DeclareSymbolFontAlphabet{\mathbbold}{bbold}

\usepackage{color}

\theoremstyle{plain}

\newtheorem{theorem}{Theorem}
\newtheorem{lemma}{Lemma}
\newtheorem{corollary}{Corollary}
\theoremstyle{definition}

\newtheorem{definition}{Definition}

\newcommand{\thmref}[1]{Theorem~\ref{#1}}

\newcommand{\E}{\mathbb{E}}

\renewcommand{\Pr}[2][]{\mbox{\rm\bf Pr}_{#1}\left[#2\right]}
\newcommand{\Ex}[2][]{\mbox{\rm\bf E}_{#1}\left[#2\right]}

\newcommand{\footremember}[2]{%
   \footnote{#2}
    \newcounter{#1}
    \setcounter{#1}{\value{footnote}}%
}
\newcommand{\footrecall}[1]{%
    \footnotemark[\value{#1}]%
} 

\begin{document}

\title{Exponential Segregation in a Two-Dimensional\\ Schelling Model with Tolerant Individuals}
\author{Nicole Immorlica\footremember{msrne}{Microsoft Research New England, Cambridge, MA.
                                   \texttt{\{nicimm,brlucier\}@microsoft.com}}
    \and Robert Kleinberg\footremember{cornellmsr}{Cornell University, Ithaca, NY; and 
                                   Microsoft Research New England, Cambridge, MA. 
                                   \texttt{rdk@cs.cornell.edu}}
    \and Brendan Lucier\footrecall{msrne}
    \and Morteza Zadimoghaddam\footremember{google}{Google Research, New York, NY. 
                                   \texttt{zadim@google.com}. Parts of this work were completed 
                                   while the author was a student at MIT.}
}
\date{}

\maketitle

\begin{abstract}
We prove that the two-dimensional Schelling segregation model yields monochromatic
regions of size exponential in the area of individuals' neighborhoods, provided that the
tolerance parameter is a constant strictly less than 1/2 but sufficiently close to it.
Our analysis makes use of a connection with the first-passage percolation model from the
theory of stochastic processes.
\end{abstract}

\section{Introduction}
\label{sec:intro}

Almost 50 years ago, a landmark paper
by the economist Thomas Schelling~\cite{schelling_dynamic_1971} 
introduced a model of residential segregation in cities 
that reshaped sociologists' understanding of the underlying
basis for that process. Schelling's paper also became a 
seminal document in the field of network science,
providing what has since been described as ``the earliest
formally studied report of rampant changes in 
networks''~\cite{Hexmoor}. Yet despite the model's broad 
influence, our understanding of its predictions has, until 
recently, been largely limited to reporting the results of simulations. 
Our objective in this paper is to provide a rigorous analysis of Schelling's model in a setting that closely mimics the topology of real cities.  We find, for the setting of parameters we study, the model predicts a large degree of segregation.

Schelling segregation is a simple and intuitively appealing
stochastic process, in which 
individuals of two colors located on a graph (typically a 
one- or two-dimensional grid) randomly shift positions to 
move away from regions in which the local density of like-colored 
individuals is below a specified ``tolerance threshold''.
Using pennies and nickels on graph paper, along with a table of 
random digits, Schelling simulated the process by hand in 
the 1960's and concluded that it almost invariably reached 
a final, stable configuration with a distinctly segregated  
pattern. Thousands of researchers have reconfirmed this 
observation in computer simulations of the model and countless
variants. The finding that segregation could arise due to
individual decisions reflecting only a weak preference for
being in the majority, though it may appear obvious in
hindsight, strongly influenced the debate about the causes
of urban segregation in the 1980's and 1990's.
\citet{clark_understanding_2008} write,
``To that point, most social scientists offered an explanation 
for segregation that invoked housing discrimination, 
principally by whites, as the major force in explaining why 
there was residential separation in the urban fabric \dots
The Schelling model was critical in providing a theoretical 
basis for viewing residential preferences as relevant to 
understanding the ethnic patterns observed in metropolitan areas.''

Rigorous quantification of the extent of segregation in
Schelling's model has lagged far behind the literature
reporting simulation results. The first theoretical
results on the model, due to \citet{young_individual_2001},
modified the segregation process by perturbing its transition
rule to induce an ergodic Markov chain. This enables
identifying the set of 
{\em stochastically stable states}, i.e., those whose stationary 
probability remains bounded away from zero as the
magnitude of the perturbation converges to zero.
Analyzing a one-dimensional version of Schelling 
segregation, Young identified the set of stochastically
stable states as those in which total segregation---partitioning
the world into two monochromatic intervals---arises.
\citet{zhang_dynamic_2004} extends this result to 
two dimensions, again finding that the stochastically
stable states are those that minimize the length of 
the interface between the two types of individuals. 

There are two reasons why the results of stochastic
stability analysis are not entirely satisfactory as
an account of segregation. First, unlike 
Schelling's original model, stochastic stability
predicts outcomes in which the pattern of segregation
differs from patterns observed in real cities,
both quantitatively (as expressed, for example, by
the fraction of people having at least one oppositely-colored
neighbor) and its extreme geometric regularity. Second,
as was observed by \citet{mobius2000formation}, the
Markov chain analyzed by Young and Zhang has a very
long mixing time; it converges much more rapidly to 
moderately-segregated {\em metastable states} which
tend to persist for an exponential number of steps before
giving way to the stationary distribution itself,
which places most of its probability mass on totally
segregated states.

The issue of mixing time in the two-dimensional case 
was considered explicitly by Bhakta, Miracle and Randall \cite{Bhakta2014}.  
Bhakta et al.\ consider an ergodic Markov chain model of
segregation motivated by the theory of spin systems from statistical physics.
Their Markov chain is a perturbed, reversible version of Schelling's dynamics,
parameterized by a temperature (i.e., noise) parameter $\lambda$.
When $\lambda$ is small, Bhakta et al.\ show that the process
mixes slowly to a limiting distribution that is essentially fully segregated, in the
sense that nearly all nodes have the same color with high probability.
However, when $\lambda$ is large, the Markov chain
mixes rapidly to a highly integrated state in which it's unlikely to
observe large regions of high bias.

In contrast to the works above, we will focus on the unperturbed
Schelling model.  
The first rigorous analysis of an unperturbed Schelling
model appeared in~\cite{BIKK}, where the authors analyzed
a stochastic process on a ring of length $n$, with each
individual's neighborhood defined as the set of sites
located within $w$ or fewer hops of its own location. 
In each step
of the dynamics, two individuals are chosen at random and
they swap positions if they are oppositely colored and the
fraction of like-colored individuals in their neighborhood
is less than the tolerance threshold, $\tau$. For
$\tau=1/2$, \citet{BIKK} showed that the average
length of the monochromatic intervals in the final,
stable configuration is bounded by $O(w^2)$; this
was subsequently improved to a tight bound of 
$O(w)$~\cite{IKK-unpublished}.

Counterintuitively, \citet{BELP14} proved that
the size of segregated neighborhoods
in one dimension increases sharply---from linear in $w$
to exponential in $w$---when the
tolerance threshold is changed from $\tau=1/2$ to $\tau=0.49$ or, 
more generally, when $0.3531 < \tau < 0.5$. Thus,
when individuals become more tolerant, the result is 
actually an {\em increase} in segregation. The 
explanation for this counterintuitive phenomenon
is that when $\tau < 1/2$, all but an exponentially
small fraction of the individuals are satisfied with
the composition of their neighborhood in the initial,
random configuration. The exponentially rare regions
of high bias serve as condensation nuclei for 
``monocultures'' that spread like wildfire, 
eventually engulfing all surrounding territory 
until their boundary abuts the boundary of an
oppositely-colored monoculture.

\begin{figure}
\begin{center}
\includegraphics[height=0.2\textheight]{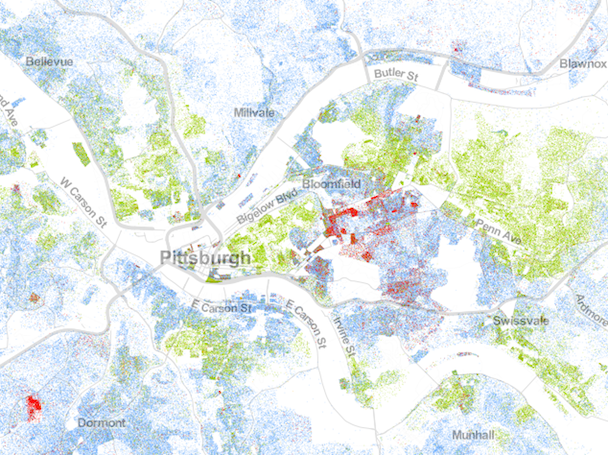}
\end{center}
\caption{\small Map of Pittsburgh, Pennsylvania, color-coded by 
race. Source: {\em The Racial Dot Map}, \citet{racial_dot_map}.}
\label{fig:housing}
\end{figure}

Understanding the behavior of the unperturbed Schelling model in two-dimensional grids, analytically, has remained a challenging open problem, but a vitally important one since the geography of urban areas is well approximated by a grid graph.  Unfortunately, the aforementioned results on the unperturbed Schelling model apply only to one-dimensional rings, and the ergodic Markov chain results, as we have argued above, do not predict the intricate patterns of segregation observed in simulations of the Schelling model as well as in real-life housing maps color-coded by ethnicity.  (See \cref{fig:housing}.)  
In a very recent unpublished manuscript, 
\citet{BELP15} achieve the first rigorous 
results on a two-dimensional, unperturbed 
Schelling model. Their results pertain to 
a model that differs from the one described
above in two key aspects. First, rather than
choosing two individuals at random and
swapping their locations if they are unhappy,
the model chooses one random individual and
changes its color if it is unhappy. In terms
of modeling segregation, this process could 
be justified, for example, by the assumption
that individuals who are unhappy with the 
ethnic composition of their neighborhood
move away from the city altogether, creating
housing vacancies that are filled using an
infinite supply of newcomers who are happy
to move to any neighborhood whose composition
satisfies their tolerance threshold. Second,
Barmpalias et al.\ assume that the two types
of individuals have potentially different
tolerance thresholds $\tau_\alpha, \tau_\beta$,
and most of their results pertain to the 
case $\tau_\alpha \neq \tau_\beta$. Both of 
these assumptions seem fairly justifiable, 
at least on intuitive grounds, when
modeling how the composition of urban neighborhoods
may change over time. In combination, however, the 
two assumptions yield results which predict either
total integration
or almost total eradication of the more tolerant population.
In more detail, when $\tau_\alpha, \tau_\beta < \frac14$,
they prove that
the fraction of nodes whose color changes during the
entire segregation process is $o(1)$, leaving most 
neighborhoods as integrated at the conclusion of the 
process as they were at its outset. 
The remaining theorems in~\cite{BELP15} 
pertain to cases in which $\tau_\alpha \neq \tau_\beta$,
and state that in those cases one color takes over 
completely, leaving at most $o(1)$ fraction of nodes
labeled with the minority color. Thus, while the
results of~\citet{BELP15} represent a breakthrough
in that they constitute the first rigorous analysis 
of an unperturbed Schelling model in two or more 
dimensions, they still fall short of providing a 
theoretical justification for the patterns of 
segregation seen in real housing data and in simulations
of Schelling's original model, which are characterized
by large monochromatic regions of both colors.

In this paper we carry out an analysis of the size of
monochromatic regions produced by the two-dimensional
Schelling segregation model when both types of individuals
have tolerance threshold $\tau < \frac12$. We assume that
the process plays out in a sufficiently large grid\footnote{
For convenience, we will assume in our analysis that the 
grid is toroidal.} populated by
individuals whose `neighborhood' is defined as the set of
all nodes at $\ell_\infty$ distance $w$ or smaller. An
individual is considered unhappy if the fraction of
like-colored individuals in its neighborhood is less
than $\tau$. In each step of the segregation process,
one individual is chosen at random and, if it is unhappy,
the label of its node changes to the opposite color. 
We assume that $\tau$ belongs to the interval 
$(\tau_0, \frac12)$ for some absolute constant $\tau_0$.
Our main result asserts that for any node,
the expected distance from that node to the 
nearest oppositely-colored node, 
when the segregation process eventually stabilizes,
is $e^{\Theta(w^2)}$. 
Thus, starting from an initially integrated
configuration, the process converges to a
pattern containing exponentially large 
segregated regions of both types.

One particularly interesting implication of our result is that
the expected size of the segregated regions in
the final configuration is non-monotone with
respect to the tolerance parameter $\tau$.
As mentioned above, Barmpalias et al. \cite{BELP15}
show that for $\tau < \frac14$, most nodes retain
their initial color and hence the expected size
of the final monochromatic regions is $O(1)$.  
On the other hand, for $\tau$ close to
$1/2$, we show that the expected size is 
$e^{\Theta((\frac12 - \tau)^2 w^2)}$, which is
exponentially large but decreasing in $\tau$.  One can therefore
find $\tau_1, \tau_2$ with $\frac14 < \tau_1 < \tau_2 < \frac12$ such that the expected
size of the final monochromatic regions is exponentially large, 
but strictly smaller for $\tau_2$ than for $\tau_1$.  
Thus, as in the single-dimensional
case \cite{BELP14}, a higher degree of tolerance can actually 
lead to an overall increase in segregation.

The intuition justifying our results is similar
to that which underlies the results 
of \citet{BELP14} in the one-dimensional case:
in the initial labeling of the grid, the 
probability of a node being unhappy is 
inverse-exponential in the area of its 
neighborhood (in this case, $\Theta(w^2)$).
When the proportion of individuals of one type
in a region becomes sufficiently high, the individuals
of the opposite type move away from that region,
starting a ``snowball effect'' whereby the region
becomes completely monochromatic, then grows to 
engulf the surrounding areas of the grid. This 
snowball effect continues unabated until the growing
monoculture encounters a region in which the
population of oppositely-colored individuals is dense
enough to halt its spread. The nearest such region is
located at a distance of $e^{\Theta(w^2)}$ from the 
point where the monoculture originated. 
See Figure~\ref{fig:sim} for an illustration.

\begin{figure}
\begin{center}
\includegraphics[height=0.17\textheight]{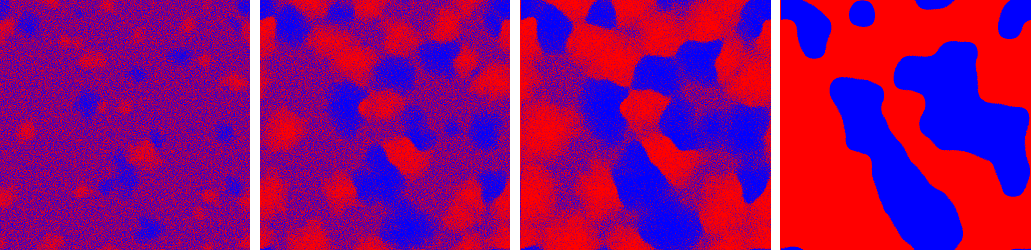}
\end{center}
\caption{\small Snapshots from a simulation of the Schelling process with $n=750$, $w = 20$, $\tau = 0.40$.}
\label{fig:sim}
\end{figure}

Although our result stems from the same intuition
as the result on exponential segregation in one
dimension, the techniques used to prove it are quite
different. Two technical 
aspects of the proof highlight key distinctions
between the one-dimensional and two-dimensional
processes: one concerns the local dynamics leading
to the genesis of the ``snowball effect'' mentioned
in the preceding paragraph, the other concerns the
global dynamics by which competing, oppositely-colored
regions undergoing this snowball effect grow to occupy
a large region. In more detail, the first distinction
between the one-dimensional and two-dimensional processes
concerns the amount of bias  (i.e., density of the
majority type)  necessary to trigger the rampant
growth of a monochromatic region.
In one dimension, \citet{BELP14} show that
any interval containing an unhappy node has
a tendency to grow larger until it encounters
an (exponentially rare) stable configuration
that can block it. In two dimensions, 
most nodes that are unhappy in the initial
configuration trigger only a bounded
number of changes to the labeling of the grid, 
all located within distance $O(w)$ of the 
original unhappy node, ending at a state in 
which all nearby nodes are happy. Instead,
the long-range propagation of segregation
is attributable to nodes whose 
neighborhood has a slightly larger amount
of bias in the initial configuration. These
``viral'' nodes are exponentially more rare
than unhappy nodes, but still occur with
sufficiently high frequency to account for
the $e^{\Theta(w^2)}$ scaling in our main
result. 
The spread of the ``monocultures'' originating at 
these viral nodes bears some similarity to a 
competitive contagion process. In one dimension
the analysis of this competitive contagion is 
aided by the simplicity of the geometry: a monoculture can only be 
blocked by encountering a blocking structure immediately to its
left or right. In two dimensions the geometry 
becomes more complicated, but we overcome this 
complexity by coupling the segregation process
with first-passage percolation (FPP), enabling
the application of powerful theorems about the
tendency of the latter process to approximate a
well-defined limiting shape.

The rest of our paper is organized as follows.
In \cref{sec:model} we formalize the Schelling
segregation process, introduce our notation, and
formally state our main result.
In \cref{sec:outline} we outline our proof strategy.
\cref{sec:viral} and \cref{sec:percolation} 
present the two main steps of the proof,
and \cref{sec:mainthm} completes the proof.

\subsection{Other Related Work}

Our analysis borrows from a literature on stochastic processes on networks that spans computer science, discrete mathematics, and probability theory.  There has been a rich line of recent work analyzing the mixing time for various local spin models on lattices \cite{Dyer2004,Lubetzky2012,Randall2013,Goldberg2005,Randall2007}.  It is often necessary in such analyses to bound the extent to which one site can exert influence on another, and our work shares a similar flavor.  
Such analyses have been used recently in computer science to derive connections to computional barriers in random instances of hard optimization problems \cite{Achlioptas2008,Sly2010}, and as the basis for new counting algorithms \cite{Jerrum1995,Bandyopadhyay2006,Bayati2007,Maneva2007}.

Our work is also related to the literature on competitive influence on networks, in which two or more processes (e.g., adoption of competing products) spread through a network, each having a suppressing effect on the other(s).  A common goal is to determine how the network structure and initial configuration influence the size of the contagions in the final state.  Such processes have been studied through the lens of strategic choice \cite{Montanari2009,Montanari2010}, as well as purely stochastic models of influence propagation \cite{Borodin2010,goyal2014competitive,Bharathi2007}.  Our analysis of the Schelling model proceeds by drawing connections with a model of epidemic spread, and bounding the extent to which one process can spread before it is affected by a competing process.

\section{Model}
\label{sec:model}

We consider a network of $n^2$ nodes arranged in an $n \times n$ two-dimensional torus.  Given any $r \geq 1$, the $r$-{\em neighborhood} of a node $x$, written $N_r(x)$, is the set of nodes $y$ with $||x-y||_{\infty} \leq r$.  Note that this is simply a square of side-length $(2r+1)$, centered at $x$.
A \emph{configuration} is an assignment of {\em spins} (or colors) to nodes, where the set of spins is $\{+1, -1\}$.  

The model is parameterized by $w \geq 1$ and $\tau \in [0,1]$.  We consider a continuous-time dynamics in which the configuration evolves over time.  Write $\sigma_t$ for the configuration at time $t$.  In the initial configuration $\sigma_0$, each node's spin is uniformly and independently chosen at random.  Write $N(x)$ for $N_w(x)$, which we will sometimes call simply the \emph{neighborhood} of $x$.  The \emph{bias} at a node $x$ at time $t$, $b_t(x)$, is the sum of spins of the nodes in $N(x)$.  That is, $b_t(x) = \sum_{y \in N(x)} \sigma_t(y)$.  
For any constant $\delta \in [0,1]$, we say that a node $x$ is \emph{$\delta$-positively-biased} (resp., \emph{$\delta$-negatively-biased}) at time $t$ if $b_t(x) > \delta|N(x)|$ (resp., $b_t(x) < - \delta|N(x)|$).  We say that a node is \emph{$\delta$-biased} if it is $\delta$-positively-biased or $\delta$-negatively-biased.

A node is said to be {\em unhappy} at time $t$ if fewer than a $\tau$ fraction of its neighbors share its color.  We will focus on the case $\tau = (1-\epsilon)/2$ where $\epsilon > 0$ is small.  Then node $x$ is unhappy precisely if it is $\epsilon$-biased and $\sigma_t(x) \cdot b_t(x) < 0$.  In our results, we will typically think of $w$ as being large relative to $1/\epsilon$, and $n$ as large relative to $w$.


Each node is assigned a Poisson clock that rings with rate $1$.  When a node's clock rings, if that node is unhappy, its spin is switched.  The dynamics ends when there are no unhappy nodes (which must occur after a finite number of switches when $\tau \leq 1/2$; see, e.g., \cite{zhang_residential_2004}).  Write $T$ for the time at which the dynamics ends, so that $\sigma_T$ is the final configuration.


A set of nodes is {\em monochromatic} at time $t$ if they all have the same spin in $\sigma_t$.
Write $M_t(x)$ for the largest monochromatic neighborhood containing node $x$ in configuration $\sigma_t$.  Write $r_t(x)$ for the radius of $M_t(x)$.  In the following statement, probabilities are taken over the initial configuration, the randomness in the dynamics, and the choice of node $x$.  We will prove:

\begin{theorem} \label{thm.stmt2}
Fix $w$, and take $n$ sufficiently large and $\epsilon > 0$ sufficiently small.  Then there exists a constant $c > 0$ such that $\Ex[x]{r_T(x)} \geq e^{c \epsilon^2 w^2}.$
\end{theorem}

We will also prove that the dependence on $\epsilon$ and $w$ in the exponent of the bound in Theorem~\ref{thm.stmt2} is asymptotically tight; the proof appears in Appendix \ref{sec:upperbound}.

\section{Proof Strategy}
\label{sec:outline}

\thmref{thm.stmt2} states that the expected radius of the largest monochromatic ball containing an arbitrary node is exponential in $w^2$.  Our proof strategy tracks the following intuition.  Since $\tau$ is bounded away from $1/2$, most nodes in the initial configuration are happy.  Unhappy nodes are exponentially unlikely, and appear in clusters that are exponentially far apart.  Some of these unhappy nodes have neighborhoods that are so biased, in the initial configuration, that they have a non-trivial probability of turning the surrounding area monochromatic as the unhappy nodes switch sign.  This, in turn, causes other nearby nodes to become unhappy and switch sign.  These ``viral'' regions will then grow, causing a larger and larger area to become heavily biased and then monochromatic.  Such a monochromatic region will spread until running up against another region that is heavily biased in the opposite sign.  Since these regions grow from initially unhappy nodes that are exponentially far apart, we expect them to grow exponentially large before running up against each other.

This intuition suggests the following steps for a formal proof:

\medskip \noindent
\textbf{1.} Define an unhappy node $x$ to be {\em viral} if its neighborhood has absolute bias at least $(\epsilon+\epsilon^2)(2w+1)^2$.  We prove that this bias is sufficiently high that $N_{w/4}(x)$ is very likely to become monochromatic within a short (polynomial in $w$) amount of time.  (Lemma \ref{lem:viral.sequence} and Lemma \ref{lem:viral.to.monochrome}.)

\medskip \noindent
\textbf{2.} We argue that the probability that a node is viral, given that it is unhappy, is significantly higher than the probability that any given node is unhappy.  Specifically, the probability that a node is unhappy in the initial configuration is at most $e^{-\Theta(\epsilon^2 w^2)}$, whereas the probability that a node has enough initial bias to be viral given that it's unhappy is at least $e^{-\Theta(\epsilon^3 w^2)}$.  (Lemma \ref{lem:viral.given.unhappy}.)

\medskip \noindent
\textbf{3.} Fix an arbitrary node $x$.  Let $R$ be chosen so that the probability there is no unhappy node in $N_R(x)$, in the initial configuartion, is approximately $1/2$.  Then $R=e^{\Theta(\epsilon^2 w^2)}$.  We will choose some $r < R$, and consider the event that there is some viral node $y \in N_r(x)$, but no unhappy node of the opposite spin lies in $N_R(x)$.  We prove that this event does not have too small a probability.  Indeed, the closest unhappy node to $x$, say $y$, is not too unlikely to be viral; $y$ will appear in $N_r(x)$ with probability $\Theta(r^2 / R^2)$; and even if we condition on $y$ being viral, this conditioning only makes it less likely that any other node in $N_R(x)$ is biased in the opposite direction. (Section \ref{sec:mainthm}.)

\medskip \noindent
\textbf{4.} Finally, we argue that if the events from the previous step come to pass, then (with high probability) node $y$ will eventually generate a monochromatic region that contains node $x$.  This proof uses a connection to first passage percolation, which bounds the rate at which an infectious process spreads in an infinite lattice.  We prove that monochromatic regions spread at a certain rate, as long as the surrounding region does not contain nodes that are sufficiently biased in the opposite color.  We also show that unhappiness can spread through the lattice at a (potentially faster) rate, again using a connection to first-passage percolation.  We must therefore settle a race condition: we prove that as long as $r$ is sufficiently smaller than $R$, then the monochromatic region will engulf node $x$ before it can be influenced by any initially unhappy node of the opposite color.  Even under a pessimistic assumptions about the rate by which viral nodes turn neighboring regions monochromatic, setting $r = R/w^3$ will be sufficient to establish that the viral region spreading from $y$ will ``reach'' node $x$ before it can be influenced by any node lying outside $N_R(x)$.  (Lemma \ref{lem:unhappy} and Lemma \ref{lem:big.ball}.)

\medskip \noindent
\textbf{5.} Putting the above pieces together, we can conclude that at some point in time, $x$ will lie in a monochromatic region of radius $\Omega(r) = e^{\Omega(\epsilon^2 w^2)}$, with probability at least $\Omega(\frac{r^2}{R^2}) \cdot e^{-O(\epsilon^3 w^2)}$, which is $e^{\Omega(\epsilon^2 w^2)}$ in expectation.  The main result then follows from the fact that such large monochromatic regions will necessarily persist until the final configuration (Lemma \ref{lem:ball.persists}).



\section{Viral Nodes}
\label{sec:viral}


We will define a node to be viral if its neighborhood is sufficiently biased.  

%
\begin{definition}
We say that a node $x$ is {\em viral} if it is $(\epsilon+\epsilon^2)$-biased in the initial configuration.
\end{definition}

The amount of bias needed to be viral is greater than the amount of bias needed to be unhappy.  The following lemma shows that this additional amount of bias is modest enough that any given unhappy node is not too unlikely to be viral.

%
%
\begin{lemma}
\label{lem:viral.given.unhappy}
There exist constants $c_1$ and $c_2$ such that
\[ \Pr[]{v \text{ is viral } | v \text{ is $\epsilon$-biased }} \geq c_1 e^{-c_2 \epsilon^3 w^2}. \]
\end{lemma}
In \cref{sec:proof-of-viral.given.unhappy} we supply the proof of the lemma,
which is an easy application of binomial tail estimates.
We next show that if a node $x$ is viral, then there is a short sequence of node activations that leads to a large area around $x$ being monochromatic.

\begin{lemma}
\label{lem:viral.sequence}
Conditional on a node $x$ being viral, with probability at least $1-2w^3e^{-\epsilon^4w}$ over the initial configuration, there exists a sequence of at most $w^2$ proposed flips such that the $(w/4)$-neighborhood of $x$ becomes monochromatic.
\end{lemma}
\begin{proof}[Proof sketch.]
We sketch the proof here; full details are given in Appendix~\ref{sec:proof-of-viral.sequence}. 

Recall that $N(y)$ denotes the set of nodes in the $w$-neighborhood of $y$.  
Let $\beta=\epsilon+\epsilon^2$, so that the bias required in the definition of a viral node $x$ is $\beta |N(x)|$.  Let $(0,0)$ be the coordinates of the viral node $x$ and suppose without loss of generality its neighborhood has positive bias.  Let $A(r)$ be the set of nodes $y$ in $N(x)$ with coordinates $(a,b)$ satisfying $|a|+|b|\leq rw$ for $r\in\{0,1/w,2/w,...,1/2\}$.  These nodes form a diamond shape centered at the origin.  In the proof, we will consider a sequence of proposed
flips arranged in a sequence of concentric diamonds; in other words, all nodes
in the set $A(r)$ are proposed to flip after the nodes in the set 
$\bigcup_{s < r} A(s)$. The crux of the argument involves conditioning on 
a set of high-probability events regarding the bias of sets of nodes in $N(x)$ or near $N(x)$ and proving that, when all of these events
occur, all of the proposed flips in the sequence are successful. The 
proof of this step is inductive: we show that for any node $y$ in the
sequence, if all of the preceding proposed flips were successful, then
the resulting bias in $N(y)$ is high enough to make $y$ unhappy. 
This involves partitioning $N(y)$ into three regions, counting lattice
points in each region, and using the events on which we conditioned
to bound the bias in each of the three regions from below. Our assumption
that the initial bias in $N(x)$ is at least 
$\beta=\epsilon+\epsilon^2$ is designed to make
this calculation work out; for significantly
smaller values of $\beta$ (e.g., $\beta = \epsilon + \epsilon^3$)
one could in fact prove the opposite result, i.e.\ that with high
probability there is no sequence of proposed flips
that turns the $(w/4)$-neighborhood of $x$ monochromatic.
\end{proof}

\section{First-Passage Percolation}
\label{sec:percolation}

%

Our proofs will use First-Passage Percolation (FPP) on a 2D lattice, where two nodes to be connected by an edge if they are horizontally, vertically, or diagonally adjacent.  In FPP, 
nodes have i.i.d.\ random weights with cumulative distribution function $F$.
Write $B(t)$ for the set of 
nodes whose shortest path to the origin (i.e., summing over nodes in the path) has length at most $t$.\footnote{Classically, First-Passage Percolation uses weights on links, rather than nodes, and does not include edges between diagonally-adjacent nodes.  However, these changes only impact the exact constants in the results below, which we are suppressing in our description.}  
Given a set $B$ of nodes and a scalar $d$, we will write $d \cdot B$ to mean the set of nodes $(x,y)$ for which $(\lfloor x/d \rfloor, \lfloor y/d \rfloor)$ lies within set $B$.
For a given radius $r$, let $D(r)$ be the set of nodes $y$ with $||y||_{\infty} \leq r$.  We wish to find bounding boxes $D(R)$ and $D(r)$, with $R > r$, such that $B(t)$ strictly contains $D(r)$ and is strictly contained in $D(R)$.  As long as $R$ and $r$ are not too far apart, these ``bounding boxes" will give reasonable bounds on the rate of growth of $B(t)$.  The following result, which is a restatement of a result due to Kesten \cite{Kesten1993}, provides such bounds, as a function of $\E[F]$.  We show how to derive this restatement in Appendix \ref{app:fpp}.

\begin{theorem}
\label{thm:fpp}
There exist fixed constants $\mu_1, \mu_2, C_1$, and $C_2$, such that the following is true.  Suppose the conditions of Theorem \ref{thm:fpp} hold, let $\lambda = \E[F]$, and let $t$ be sufficiently large (i.e., larger than a certain fixed constant).  Then
\[ \Pr[]{ B(t) \subseteq D( \mu_2 t/\lambda ) } \geq 1 - e^{-C_1 (t/\lambda)^{1/2}}, \text{ and } \]
\[ \Pr[]{ D( \mu_1 t / \lambda ) \subseteq B(t) } \geq 1 - e^{-C_2 (t/\lambda)^{3/8}}. \]
\end{theorem}

We will now employ first-passage percolation to show that if a node $y$ is very far from any unhappy node with a negative bias, then one can bound the probability that node $y$ changes spin from positive to negative before a certain amount of time has passed.

\begin{lemma}
\label{lem:unhappy}
There exist constants $\mu_2$ and $C_1$ such that the following is true.  Fix any $R > 0$, choose any node $x$, and suppose that, in the initial configuration $\sigma_0$, there are no $\epsilon$-negatively-biased nodes contained in $N_R(x)$.  Then for $t_1 = R / (8 w^3 \mu_2)$, the probability that there is any $\epsilon$-negatively-biased node within $N_{R/4}(x)$ at any time $t \leq t_1$ is at most $8R e^{-C_1 t_1^{1/2} w}$.
\end{lemma}
\begin{proof}
The proof will proceed by comparing the Schelling process to a certain FPP process.  Partition the lattice into blocks of width $w$.  We will say that a block is \emph{infected} at time $t$ if, at any point at or before time $t$, a node in the block switched from positive spin to negative spin.  We'll say that a block is \emph{unhappy} at time $t$ if, at any point at or before time $t$, any node in the block is $\epsilon$-negatively-biased.  From the condition of the lemma, no block contained in $N_R(x)$ is unhappy at time $0$.  Moreover, a block can become unhappy at time $t$ only if a node in an adjacent block becomes infected at time $t$ (where diagonally-adjacent blocks are considered adjacent).  Finally, a block can become infected only if it is unhappy and a node in the block is selected to update; since there are $w^2$ nodes in a block, the time delay before a node in a given block is selected to update is distributed like an exponential random variable with expectation $1/w^2$.

The above discussion immediately implies that the set of infected blocks at time $t$ is stochastically dominated by the set of active blocks in a FPP process over blocks, where the initial seeds are the initially unhappy blocks, and the node weights are distributed as exponential random variables with expectation $1/w^2$.   

It suffices to argue that no block intersecting $N_{R/2}(x)$ is activated by this FPP process by time $t_1$, as this implies that no block intersecting $N_{R/4}(x)$ is ever adjacent to an infected block before time $t_1$, and hence no block in $N_{R/4}(x)$ is unhappy during that time period.  

The worst initial configuration, subject to the conditions of the lemma, occurs if all $8R$ nodes at distance $R$ from $x$ are unhappy.  Note that the distance between any of these nodes and $N_{R/2}(x)$ is at least $4w^3\mu_2 t_1$.  For any block $V$ containing a node at distance $R$ from $x$, the probability that the FPP process beginning at block $V$ activates a block beyond distance $2 w^2 \cdot w \mu_2 t_1$ from $V$ is at most $e^{-c t_1^{1/2} w}$ for some constant $c$, by Theorem \ref{thm:fpp}.  (Here we used $\lambda = 1/w^2$, and scaled distances by $w$.)  Thus, taking a union bound over all $8R$ initially unhappy nodes, the probability that any block within $N_{R/2}(x)$ becomes activated is at most
$8R \cdot e^{-C_1 t_1^{1/2} w}$
as required.
\end{proof}

We are now able to show that Lemma \ref{lem:viral.sequence} and Lemma \ref{lem:unhappy} together imply that, with high probability, a viral node is likely to generate a monochromatic region.  The proof appears in Appendix \ref{app:viral.to.monochrome}; the idea is to use Lemma \ref{lem:unhappy} to argue that no nearby nodes will biased toward the opposite color near the viral node, and hence it's likely that the sequence from Lemma \ref{lem:viral.sequence} will actually occur and generate a monochromatic region.

\begin{lemma}
\label{lem:viral.to.monochrome}
There exists a constant $c > 0$ such that the following is true.  Fix $r > 8 w^6 \mu_2$, choose any node $x$, suppose node $y \in N_r(x)$ is viral with positive bias, and furthermore there are no $\epsilon$-negatively-biased nodes in $N_{8r}(x)$.  Then, with probability at least 
$1 - e^{- c \epsilon^4 w}$,
the neighborhood $N_{w/4}(y)$ is positive monochromatic at time $t_2 = w^3$.
\end{lemma}

Finally, we will again use first-passage percolation to argue that once a viral node has a monochromatic neighborhood, this monochromatic region will grow at a sufficiently fast rate, up until the point where it may encounter a node that is unhappy with the opposite bias.


\begin{lemma}
\label{lem:big.ball}
There exist constants $c > 0$ and $\mu_1 > 0$ such that the following is true.
Fix any $r > w^3$ and node $x$, and suppose that node $y \in N_r(x)$ is viral with positive bias.
Set $t_3 = w^3 + 2 \log(w) r / \mu_1$, and suppose that at all times $t < t_3$, there is no $\epsilon$-negatively-biased node within $N_{2r}(x)$.  Then at time $t_3$, $x$ is contained in a monochromatic region of radius $r$ with probability at least $1 - e^{ - c \epsilon^4 w}$.
\end{lemma}
\begin{proof}
By Lemma \ref{lem:viral.to.monochrome}, we can consider the event that, by time $w^3$, the neighborhood $N_{w/4}(x)$ is positively monochromatic; this event occurs with probability at least $1 - e^{-\Theta(\epsilon^4 w)}$.
Partition the lattice into blocks of width $w/4$, aligned with $N_{w/4}(x)$.  By assumption, no node is $\epsilon$-negatively-biased at any time before $t_3$, within $N_{2r}(x)$.  

Consider the event that, in the initial configuration, there is a block contained in $N_{2r}(x)$ that has more than $3/4$ of its nodes assigned a positive spin.  Chernoff bounds imply that this occurs with probability at most $e^{-c_1 w^2}$ for some constant $c_1$.  For the remainder of the proof, we will condition on this event not occurring.

Our proof will proceed by comparing the spread of monochromatic blocks to a FPP process.  Consider any block contained in $N_{2r}(x)$, and consider the event that this block becomes monochromatically positive before time $t_3$.  Note that since no block has more than $3/4$ of its nodes having a positive spin in $\sigma_0$, it must be that at least $(\frac{1}{4})w^2/16$ nodes in the block changed sign from negative to positive in such a change.  This will shift the bias of any node in an adjacent block by at least $w^2 / 128$.  Assuming $\epsilon < 1/128$, this implies (since no node was $\epsilon$-negatively-biased in the initial configuration) that every negatively-signed node in an adjacent block must be unhappy after such a change.  Thus, each such node would subsequently switch to a positive sign if it were selected.  Let $F$ be the cumulative distribution function for the amount of time until every node in a given block has been selected.  Note that a coupon-collector argument implies that the time required for every node in a block to be selected is distributed like a sum of independent 
$k=w^2/16$ exponential random variables with rates $\frac{1}{k}, \frac{1}{k-1}, \frac{1}{k-2}, \ldots$.
We therefore have $\E[F] = H_k < 2 \ln w$.

The set of monochromatic blocks at time $t$ therefore stochastically dominates the set of active blocks in a FPP process over blocks, starting at time $w^3$.  The initial seed is the block containing $y$, and the node weights are distributed like $F$.  Here we are using the fact that $F$ is a sum of exponential random variables; in our coupling, we imagine waiting until a given block is monochromatic, then starting the (memoryless) clocks for nodes in an adjacent block, and coupling the weight of that block with the amount of time needed for all of its nodes to be selected.

Invoking the second half of Theorem \ref{thm:fpp}, and noting that $\E[F] = \Theta(\log(w))$, we have that the probability that the activation region contains a ball of radius $r$ by time $t_4 = 2 \log(w) r / \mu_1$ is at least $1 - e^{-C_2 t_4^{3/8} \log(w)^{-3/8}} \geq 1 - e^{-C_2 w}$ (where the second inequality follows from the assumption that $r > w^3$).  Since such a ball contains node $x$, taking a union bound over the failure events yields the desired result.
\end{proof}


%
%
%
%

\section{Completing the Proof of \thmref{thm.stmt2}}
\label{sec:mainthm}

Lemma \ref{lem:big.ball} bounds the probability that a given node $x$ will be contained in a large monochromatic region, at \emph{some} time $t$.  However, our main result requires that $x$ be contained in a large monochromatic region in the final configuration $\sigma_T$.  In this section, we show that a sufficiently large monochromatic region will persist until time $T$.
Given a node $x$, write $B_R(x)$ for the set of nodes $y$ with $||x-y||_2 \leq R$, where $||x-y||_2$ denotes Euclidean distance.

\begin{lemma}
\label{lem:ball.persists}
Fix $\epsilon > 0$ and take $w$ sufficiently large.  Suppose $B_R(x)$ is monochromatic in configuration $\sigma_t$, where $R \geq w^3$.  Then for any $t' > t$, $B_R(x)$ is monochromatic in $\sigma_{t'}$.
\end{lemma}

The proof appears in Appendix \ref{app:ball.persists}.  We now have the pieces necessary to complete the proof of \thmref{thm.stmt2}, following the outline presented in Section \ref{sec:outline}.  The details appear in Appendix \ref{app:mainthm}.

\bibliographystyle{apalike}
\bibliography{biblio}

\appendix
\section{Proof of \cref{lem:viral.given.unhappy}}
\label{sec:proof-of-viral.given.unhappy}

\cref{lem:viral.given.unhappy} is readily seen to follow
from the following lemma about binomial tails and its
corollary.

\begin{lemma} 
\label{lem:binomial-tail}
Let $X$ be a random variable with the binomial distribution
$B(n,\frac12)$. Then for any $q > r > \frac{n}{2}$,
\begin{equation}
\label{eq:bt}
\Pr{X \geq q \,|\, X \geq r} > \frac{1}{n} \cdot 
\left( \tfrac{n-q}{r} \right)^{q-r}.
\end{equation}
\end{lemma}
\begin{proof}
We have
\begin{equation*}
\Pr{X \geq q \,|\, X \geq r} =
\tfrac{\Pr{X \geq q}}{\Pr{X \geq r}} =
\tfrac{\sum_{i=q}^{n} \binom{n}{i}}{\sum_{i=r}^{n} \binom{n}{i}} >
\tfrac{\binom{n}{q}}{n \cdot \binom{n}{r}}.
\end{equation*}
Now,
\begin{align*}
\tfrac{\binom{n}{q}}{\binom{n}{r}} &=
\tfrac{n!}{q!(n-q)!} \cdot \tfrac{r! (n-r)!}{n!} \\ &=
\tfrac{r!}{q!} \cdot \tfrac{(n-r)!}{(n-q)!} \\ &=
\prod_{i=1}^{q-r} \tfrac{n-q+i}{r+i} \\ &>
\left(\tfrac{n-q}{r}\right)^{q-r},
\end{align*}
which completes the proof of the lemma.
\end{proof}
\begin{corollary}
\label{cor:binomial-tail}
Let $X$ be a random variable with the binomial distribution
$B(n,\frac12)$. Then 
for any constants $\gamma,\epsilon$ such that $0 < \epsilon < \gamma < \frac13$,
\begin{equation*}
\Pr{X \geq (\tfrac12 + \gamma) n \,|\,
    X \geq (\tfrac12 + \epsilon) n} \geq
\tfrac{1}{n} \cdot 
e^{-6 (\gamma^2 - \epsilon^2) n}.
\end{equation*}
\end{corollary}
\begin{proof}
Setting $q = \lceil \tfrac12 + \gamma \rceil n$ and
$r = \lfloor \tfrac12 + \epsilon \rfloor n$, we apply
\cref{lem:binomial-tail} to obtain
\begin{equation} \label{eq:cor-bt.1}
\Pr{X \geq (\tfrac12 + \gamma) n \,|\,
    X \geq (\tfrac12 + \epsilon) n} \geq \tfrac{1}{n} \cdot 
\left( \tfrac{1 - 2\gamma}{1 + 2\epsilon} \right)^{(\gamma - \epsilon) n}.
\end{equation}
We also have
\begin{equation} \label{eq:cor-bt.2}
\tfrac{1-2\gamma}{1+2\epsilon} = 
\left( 1 + \tfrac{2(\gamma + \epsilon)}{1-2\gamma}\right)^{-1} >
\exp \left( - \tfrac{2(\gamma+\epsilon)}{1-2\gamma} \right) >
e^{-6 (\gamma+\epsilon)}.
\end{equation}
The corollary follows by combining~\eqref{eq:cor-bt.1}
with~\eqref{eq:cor-bt.2}.
\end{proof}

\cref{lem:viral.given.unhappy} follows immediately
from \cref{cor:binomial-tail} upon setting
$\gamma = \epsilon+\epsilon^2$, $n = w^2$, and assuming
that $w$ is large enough that $e^{-\epsilon^3 w^2} < \frac{1}{w^2}$.

\section{Proof of~\cref{lem:viral.sequence}}
\label{sec:proof-of-viral.sequence}

In this section we restate and prove~\cref{lem:viral.sequence}.  We need a generalized version of Chernoff bounds which can be applied on a set of negatively correlated (and not independent) random variables. The following theorem, Theorem 1.1 of \cite{Impagliazzo-Kabanets-2010}, was originally proved in \cite{Panconesi-Srinivasan-1997}.   

\begin{theorem}\label{thm:generalized.chernoff}
(Generalized Chernoff Bound \cite{Panconesi-Srinivasan-1997}, and \cite{Impagliazzo-Kabanets-2010}) Let $X_1, X_2, \cdots, X_n$ be Boolean random variables such that, for some $0 \leq \delta \leq 1$, we have that, for every subset $S \subseteq [n]$, $Pr[\land_{i \in S} X_i = 1] \leq \delta^{\lvert S \rvert}$. Then, for any $0 \leq \delta \leq \gamma \leq 1$, $Pr[\sum_{i=1}^n X_i \geq \gamma n] \leq e^{-2n(\gamma-\delta)^2}$.
\end{theorem}

\noindent Now we can use this concentration bound to prove~\cref{lem:viral.sequence}.

\begin{figure}
\begin{center}
\includegraphics[height=0.2\textheight]{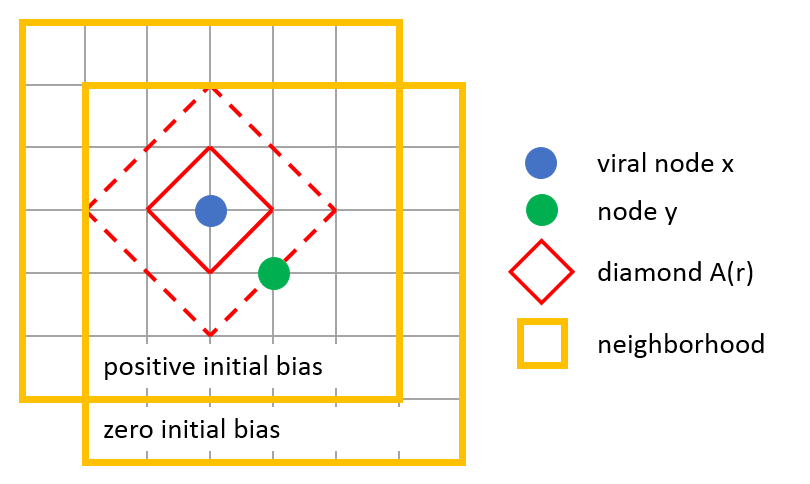}
\end{center}
\caption{\small An illustration of the regions discussed in the proof of Lemma \ref{lem:viral.sequence}.}
\label{fig:lemma2}
\end{figure}

\begin{lemma}
Conditional on a node $x$ being viral, with probability at least $1-2w^3e^{-\epsilon^4w}$ over the initial configuration, there exists a sequence of at most $w^2$ proposed flips such that the $(w/4)$-neighborhood of $x$ becomes monochromatic.
\end{lemma}
\begin{proof}
Recall that $N(y)$ denotes the set of nodes in the $w$-neighborhood of $y$.  
Let $\beta=\epsilon+\epsilon^2$, so that the bias required in the definition of a viral node $x$ is $\beta |N(x)|$.  Let $(0,0)$ be the coordinates of the viral node $x$ and suppose without loss of generality its neighborhood has positive bias.  Let $A(r)$ be the set of nodes $z$ in $N(x)$ with coordinates $(a,b)$ satisfying $\lvert a \rvert + \lvert b \rvert \leq rw$ for $r\in{0,1/w,2/w,...,1/2}$.  These nodes form a diamond shape centered at the origin.  Throughout this proof, we will condition on the events that the bias of certain sets of nodes in the initial configuration are close to their expectations.  
\begin{itemize}
\item First consider nodes $(N(y)\cap N(x))\setminus A(r)$ for $y\in N(x)$ and $r\in\{0,1/w,2/w,\ldots,1/2\}$.  Conditional on $x$ being viral, the expected fraction of these nodes with $+1$ spin is at least $1/2+\beta$.  We condition on the events that the fraction is at least $1/2+\beta-\epsilon^2/10$ for all relevant $y$ and $r$.  
By \thmref{thm:generalized.chernoff}, and the observation that $|(N(x)\cap N(y))\setminus A(r)|\geq w^2$ for all relevant $y$ and $r$, we prove that the probability that such an event doesn't hold is at most $e^{-\epsilon^4w^2/100}$. 
Set $\delta = 1/2-\beta$, and for each node in $(N(x)\cap N(y))\setminus A(r)$, consider a binary random variable which is set to $1$ if and only if the spin of the node is $-1$ in the initial configuration. 
For any subset of nodes $S \subseteq (N(x)\cap N(y))\setminus A(r)$, noting that the spins are negatively correlated each with marginal probability of $\delta$ of being $-1$, the probability that $S$ itself is negative monochromatic in the initial configuration is at most $\delta^{\lvert S \rvert}$. For $\gamma = \delta + \epsilon^2/10$, if the number of $-1$ spins is at most $\gamma|(N(x)\cap N(y))\setminus A(r)|$, then this implies the desired lower bound on the fraction of $+1$ spin nodes.
We apply the union bound for all $y \in N(x)$, and $r\in\{0,1/w,2/w,\ldots,1/2\}$, to conclude that with probability at least $1-w^3e^{-\epsilon^4 w^2/100}$, the fraction of nodes with $+1$ spin in $(N(x)\cap N(y))\setminus A(r)$ is at least $1/2+\beta-\epsilon^2/10$. 
\item Similarly, we condition on the events that the negative bias of nodes $N(y) \setminus N(x)$ for $y\in N(x)$ is at most $\epsilon^2 |N(y) \setminus N(x)|$.   As the expected bias of these nodes is $0$, and $|N(y) \setminus N(x)|>2w$ for all $y$, the probability for each $y$ that the negative bias is more than $\epsilon^2|N(y) \setminus N(x)|$ is at most $e^{-\epsilon^4 w/2}$ by Chernoff bounds. Again, by the union bound, the probability that all these events occur is at least $1-4w^2e^{-\epsilon^4 w/2}$.  
\end{itemize}
As events concerning nodes in $N(x)$ and nodes outside $N(x)$ are independent, we conclude that all these events hold with probability at least $(1-w^3e^{-\epsilon^4 w^2/100})(1-4w^2e^{-\epsilon^4 w/2})$, which is greater than $1-2w^3 e^{-\epsilon^4w/2}$ provided that $w$ is sufficiently large.

We consider a sequence in which all nodes in set $A(r)$ are proposed to flip before nodes in set $A(r+1/w)$ and prove by induction that, conditional on the events above, each node in this sequence prefers to adopt spin $+1$.  When $r=0$, $A(r)=\{x\}$ which, by the conditions of the lemma, prefers to adopt spin $+1$.  For $r>0$, suppose all nodes in set $A(r)$ have already adopted spin $+1$.  Let $y$ be a node in $A(r+1/w)-A(r)$.  Note that since $r+1/w\leq 1/2$, $A(r)\subset N(y)$.  The number of nodes in $A(r)$ is $2rw(rw+1) + 1$.  Nodes in the neighborhood of $y$ can be partitioned into three sets:
\begin{enumerate}
\item Nodes in $N(y)\cap A(r)=A(r)$.  As discussed above, the number of such nodes is
$2rw(rw+1)+1\geq 4w^2(r^2/2)$ and these nodes all have spin $+1$ by the inductive hypothesis for a total contribution to the bias of $N(y)$ of $$4w^2(r^2/2).$$
\item Nodes in $(N(y)\cap N(x)) \setminus A(r)$.  As $y=(a,b)$ satisfies $\lvert a \rvert + \lvert b \rvert \leq wr+1$, the number of nodes in $N(y)\cap N(x)$ is $(2w+1-a)(2w+1-b)\geq 4w^2(1-r/2 - 1/(2w))$, and as $A(r)\subset N(y)$, the number of nodes in $(N(y)\cap N(x)) \setminus A(r)$ is at least $4w^2(1-r/2-1/(2w)-(r^2/2))$. By our conditioning above, the contribution of these nodes to the bias of $N(y)$ is at least $$4w^2 \left( 1- \tfrac{r}{2}-\tfrac{1}{2w}-\tfrac{r^2}{2} \right) \left(\beta- \tfrac{\epsilon^2}{10} \right) = 4w^2 \left( 1 - \tfrac{r}{2} -\tfrac{r^2}{2} \right) \left(\beta - \tfrac{\epsilon^2}{10} \right) - 3 \epsilon w.$$
\item Nodes in $N(y) \setminus N(x)$.  The number of such nodes is $a(2w+1-b)+b(2w+1-a)-ab<2rw^2 + (2+r)w + 1$,  and by our conditioning above, the negative bias of this set is at most $\epsilon^2$ times its area.  Thus this set decreases the positive bias of $N(y)$ by at most $$4w^2 \left( \tfrac{r}{2} + \frac{3}{4w} \right)\epsilon^2 = 4w^2 \left( \tfrac{r}{2} \right) \epsilon^2 + 3 \epsilon^2 w.$$
\end{enumerate}
In sum, the total positive bias in the $w$-neighborhood of $y$ is at least:
\begin{equation} \label{eq:viral.1}
4w^2\left[\frac{r^2}{2}\left(1-\beta+\frac{\epsilon^2}{10}\right)-\frac{r}{2}\left(\beta+\frac{9\epsilon^2}{10}\right)+ \left(\beta-\frac{\epsilon^2}{10}\right)\right]
- 3 \beta w.
\end{equation}
For $y$ to prefer $+1$ spin, we need to argue that the above expression is more than $4w^2\epsilon$.  The minimum occurs at $$r=\frac{\beta+9\epsilon^2/10}{2(1-\beta+\epsilon^2/10)}.$$  At this setting the quadratic equation becomes:
$$\beta-\frac{\epsilon^2}{10}-\frac{(\beta+9\epsilon^2/10)^2}{8(1-\beta+\epsilon^2/10)}.$$
Recall that $\beta=\epsilon+\epsilon^2$.  This implies that $\beta<2\epsilon$ and hence $\beta^2<4\epsilon^2$.  Choose $\epsilon$ small enough that $8(1-\beta+\epsilon^2/10)>6$. Substituting into the equation, we see that the bias is at least:
\begin{eqnarray*}
\beta-\frac{\epsilon^2}{10}-\frac{\beta^2+(9/5)\beta\epsilon^2+(81/100)\epsilon^4}{8(1-\beta+\epsilon^2/10)}&\geq&
\epsilon+\frac{9\epsilon^2}{10}-\frac{5\epsilon^2}{6}>\epsilon + \frac{\epsilon^2}{15}.
\end{eqnarray*}
For $w > 45\beta/(\epsilon^2)$ the excess bias
$\frac{1}{15} w^2 \epsilon^2$ in this calculation exceeds the
remainder $3 \beta w$ in~\eqref{eq:viral.1}.
\end{proof}

\section{First-Passage Percolation: Details}
\label{app:fpp}

Our proofs use First-Passage Percolation (FPP) on a 2D lattice, where two nodes to be connected by an edge if they are horizontally, vertically, or diagonally adjacent.  In FPP, 
nodes have i.i.d.\ random weights with cumulative distribution function $F$.
Write $B(t)$ for the set of 
nodes whose shortest path to the origin (i.e., summing over nodes in the path) has length at most $t$.\footnote{Classically, First-Passage Percolation uses weights on links, rather than nodes, and considers does not include edges between diagonally-adjacent nodes.  However, these changes only impact the exact constants in the results below, which we are suppressing in our description.}  
Given a set $B$ of nodes and a scalar $d$, we will write $d \cdot B$ to mean the set of nodes $(x,y)$ for which $(\lfloor x/d \rfloor, \lfloor y/d \rfloor)$ lies within set $B$.
For a given radius $r$, let $D(r)$ be the set of nodes $y$ with $||y||_{\infty} \leq r$.  We wish to find bounding boxes $D(R)$ and $D(r)$, with $R > r$, such that $B(t)$ strictly contains $D(r)$ and is strictly contained in $D(R)$.  As long as $R$ and $r$ are not too far apart, these ``bounding boxes" will give reasonable bounds on the rate of growth of $B(t)$.  The following result, which is a restatement of a result due to Kesten \cite{Kesten1993}, provides such bounds, as a function of $\E[F]$.  We show how to derive this restatement from Kesten's original theorem in Appendix \ref{app:fpp}.

The following is a restatement of a result due to Kesten \cite{Kesten1993}:

\begin{theorem}[\cite{Kesten1993}]
\label{thm:fpp.original}
Suppose $\E[F] < \infty$ and furthermore $\int e^{\gamma x}F(dx) < \infty$ for some $\gamma > 0$.  Then there exists a compact set $B_0$ (depending on $F$) and fixed constants $C_1, \dotsc, C_5$, such that the following are true:
\[ \Pr[]{B(t) \subseteq t (1 + x/\sqrt{t}) B_0 } \geq 1 - C_1 t^4 e^{-C_2 x } \quad \text{for all $x \leq \sqrt{t}$} \]
and
\[ \Pr[]{t (1 - C_3 t^{-1/8} \log^{1/4}(t)) B_0 \subseteq B(t) } \geq 1 - C_4 t^2 e^{-C_5 t^{3/8} \log^{1/4}(t) }. \]
\end{theorem} 

Theorem \ref{thm:fpp} establishes that the active region $B(t)$ grows at a roughly linear rate, and as $t$ grows large it approximates a scaled version of a compact set $B_0$.
Actually, the result in \cite{Kesten1993} is more general; the statement of Theorem \ref{thm:fpp} restricts attention to the 2D lattice and finite $t$.

The following reformulation of Theorem \ref{thm:fpp} will be particularly useful.  For a given radius $r$, let $D(r)$ be the set of nodes $y$ with $||y||_{\infty} \leq r$.  We wish to find bounding boxes $D(R)$ and $D(r)$, with $R > r$, such that $B(t)$ strictly contains $D(r)$ and is strictly contained in $D(R)$.  As long as $R$ and $r$ are not too far apart, these ``bounding boxes" will give reasonable bounds on the rate of growth of $B(t)$.  The following Corollary provides such bounds, as a function of $\E[F]$.

\begin{corollary}
\label{cor:fpp.modified}
There exist fixed constants $\mu_1, \mu_2, C_1$, and $C_2$, such that the following is true.  Suppose the conditions of Theorem \ref{thm:fpp} hold, let $\lambda = \E[F]$, and let $t$ be sufficiently large (i.e., larger than a certain fixed constant).  Then
\[ \Pr[]{ B(t) \subseteq D( \mu_2 t/\lambda ) } \geq 1 - e^{-C_1 (t/\lambda)^{1/2}}, \text{ and } \]
\[ \Pr[]{ D( \mu_1 t / \lambda ) \subseteq B(t) } \geq 1 - e^{-C_2 (t/\lambda)^{3/8}}. \]
\end{corollary}
\begin{proof}
By rescaling, one can use $F (t / \lambda)$ rather than $F(t)$ and replace $t$ by $t / \lambda$ in Theorem \ref{thm:fpp.original}.  
Let $B_0$ be the compact set from Theorem \ref{thm:fpp} corresponding to distribution $F(t/\lambda)$, and choose $\mu_1, \mu_2$ so that $D(2\mu_1) \subseteq B_0 \subseteq D(\tfrac{1}{2}\mu_2)$.
The first statement then follows by taking $x = \sqrt{t}$ in the first inequality of Theorem \ref{thm:fpp.original} and using $B_0 \subseteq D(\tfrac{1}{2}\mu_2)$.  The second statement follows by assuming $t$ is large enough that $C_3 t^{-1/8} \log^{1/4}(t) \leq 1/2$, applying the second inequality of Theorem \ref{thm:fpp.original}, and using $D(2\mu_1) \subseteq B_0$.  The simplified forms of the RHS probabilities follow by taking $t$ sufficiently large and setting constants appropriately.
\end{proof}

\section{Proof of Lemma \ref{lem:viral.to.monochrome}}
\label{app:viral.to.monochrome}

Recall the statement of the lemma.  There exists a constant $c > 0$ such that the following is true.  Fix $r > 8 w^6 \mu_2$, choose any node $x$, suppose node $y \in N_r(x)$ is viral with positive bias, and furthermore there are no $\epsilon$-negatively-biased nodes in $N_{8r}(x)$.  Then, with probability at least 
$1 - e^{- c \epsilon^4 w}$,
the neighborhood $N_{w/4}(y)$ is positive monochromatic at time $t_2 = w^3$.

\begin{proof}
By Lemma \ref{lem:viral.sequence}, there is a sequence of $k \leq w^2$ updates that lead to the neighborhood of $y$ becoming monochromatic with probability at least $1-e^{-\Theta(\epsilon^4 w)}$.  Call this sequence of nodes $v_1, \dotsc, v_{k}$.  Since $t_2 < r / w^3 \mu_2$, Lemma \ref{lem:unhappy} implies the probability that any node within $N_{2r}(y)$ will switch spin from positive to negative, before time $t_2$, is at most $64 r e^{-C_1 t_2^{1/2}w}$.  Excluding that event, node $v_1$ will take a positive spin the first time its clock rings (say $s_1$), and then node $v_2$ will take a positive spin the first time its clock rings after $s_1$, say at time $s_2 > s_1$, and so on for each $v_i$.  This is because each node's bias can only be more positive than in the precise update sequence $(v_i)$, and by Lemma \ref{lem:viral.sequence} such bias is sufficient for each node to switch to a positive spin.  The time until every node has switched, in sequence, is distributed as the sum of $k$ exponential random variables, each with expectation $1$.  Since $k \leq w^2$, standard concentration bounds imply that the probability that this sum is greater than $w^3$ is at most $e^{-c (w^3/k) k} = e^{-c w^3}$ for some constant $c$.  Taking the union bound over these bad events, we can conclude that the required sequence of flips occurs by time $t_2$ with the required probability, and after this sequence $N_{w/4}(y)$ is monochromatic, as required.
\end{proof}

\section{Proof of Lemma \ref{lem:ball.persists}}
\label{app:ball.persists}

Recall the statement of the lemma:  fix $\epsilon > 0$ and take $w$ sufficiently large.  Suppose $B_R(x)$ is monochromatic in configuration $\sigma_t$, where $R \geq w^3$.  Then for any $t' > t$, $B_R(x)$ is monochromatic in $\sigma_{t'}$.

\begin{proof}
It suffices to show that every node in $B_R(x)$ would be happy even if all nodes outside $B_2(R)$ take the opposite spin.  We will prove this by establishing that, for every $v \in B_R(x)$, at least $(\frac{1 - \epsilon}{2})(2w+1)^2$ neighbors of $v$ lie inside $B_R(x)$.  

Consider a point $v$ lying on the boundary of $B_R(x)$.  Consider the line $L$ tangent to the boundary of $B_R(x)$ at $v$.  By symmetry, precisely half of the neighbors of $v$ (excluding those that lie on $L$) lie on either side of line $L$.  It therefore suffices to show that the number of neighbors of $v$ that lie strictly between $L$ and the boundary of $B_R(x)$ is less than $\frac{\epsilon}{2}(2w+1)^2$.  

Since the boundary of $B_R(x)$ has curvature $R^{-1} \leq w^{-3}$, and the neighborhood of $v$ is contained within a ball of radius $\theta(w)$, the boundary and $L$ differ in slope by at most $\theta(w^{-2})$ within the neighborhood of $v$.  The area between $v$ and $B_R(x)$ is therefore contained within the area between $L$ and line $L$ shifted orthogonally by a distance of $\theta(1/w)$, again because the neighborhood of $v$ is contained within a ball of radius $\theta(w)$.  If $w$ is sufficiently large, then for any point $x$ in that area, the box $x + [\pm 0.5, \pm 0.5]$ must intersect line $L$.  Since any line intersects at most $\theta(w)$ such boxes, we conclude that at most $\theta(w)$ vertices lie strictly between $L$ and the boundary of $B_R(x)$.  This is less than the required $\frac{\epsilon}{2}(2w+1)^2$, for $w$ sufficiently large.
\end{proof}

\section{Completing the Proof of Theorem \ref{thm.stmt2}}
\label{app:mainthm}

Recall the statement:

\begin{theorem}[Restatement of Theorem \ref{thm.stmt2}]
Fix $w$, take $n$ sufficiently large, and take $\epsilon > 0$ sufficiently small.  Then there exists a constant $c > 0$ such that $\Ex[x]{r_T(x)} \geq e^{c \epsilon^2 w^2}.$
\end{theorem}

\begin{proof}
Pick an arbitrary node $x$.  For any other node $z$, Chernoff bounds imply that the probability $z$ is unhappy in the initial configuration is $e^{-c_1 \epsilon^2 w^2}$ for some constant $c_1 > 0$.   Choose any $R$ such that, with probability lying in $[1/3, 2/3]$, no node in $N_R(x)$ is $\epsilon$-biased.  This will occur if the number of nodes lying in $N_R(x)$ is $e^{\Theta(\epsilon^2 w^2)}$, so we have $R = e^{c_2 \epsilon^2 w^2}$ for some constant $c_2$.

We now wish to condition on a certain sequence of events.  Set $r = R / w^4$.  First, we will condition on there existing at least one $\epsilon$-biased node in $N_{r/2}(x)$.  This occurs with probability $\Theta(r^2 / R^2) = \Theta(1 / w^8)$.  Let $y$ be the closest such node to $x$.  Assume, by symmetry, that $y$ is positively biased.  We will then condition on node $y$ being viral.  By Lemma \ref{lem:viral.given.unhappy}, this occurs with probability at least $c_3 e^{-c_4 \epsilon^3 w^2}$ for constants $c_3, c_4 > 0$, conditional on $y$ being $\epsilon$-biased.  Finally, we will condition on there not existing any $\epsilon$-negatively-biased nodes within $N_R(x)$.  Since conditioning on the positive bias of $y$ can only make the distribution of bias of any other node more positive, this event has conditional probability at least $1/3$, from our choice of $R$.

Let $t_0 = R / (8w^3\mu_2) = \frac{1}{8w^3\mu_2} e^{c_2 \epsilon^2 w^2}$.  By Lemma \ref{lem:unhappy}, there is no negatively-biased node contained in $N_{R/4}(x)$ at any time $t \leq t_0$, with probability at least $1 - 8R e^{- \Theta(t_0^{1/2} w)} = 1 - e^{-e^{\Theta(\epsilon^2 w^2)}}$.  Since $R / 4 > 2r$ and $t_0 = r w^4 / (8 w^3 \mu_2) > w^3 + 2 \log(w) r / \mu_1$ for sufficiently large $w$, Lemma \ref{lem:big.ball} implies that, with probability at least $1 - e^{-c_5 \epsilon^4 w}$ for some constant $c_5 > 0$, node $x$ is contained in a monochromatic region of radius at least $r$, centered at $y$, at some time $t \leq T$.  Since $y \in N_{r/2}(x)$ implies $||x-y||_2 \leq r$, Lemma \ref{lem:ball.persists} then implies that the origin is contained in a monochromatic square of radius at least $r/2$ in the final configuration.

Taking the union bound over all excluded events, we conclude that with probability at least $e^{-c_6 \epsilon^3 w^2}$, at time $T$ we will have $r_T(x) \geq r = e^{c_7 \epsilon^2 w^2}$ for some constants $c_6,c_7 > 0$.  We therefore have $\Ex[]{r_T(x)} \geq e^{\epsilon^2(c_7  - c_6 \epsilon) w^2}$, which yields the desired result for sufficiently small $\epsilon$.
\end{proof}

\section{An Upper Bound on Monochromatic Region Size}
\label{sec:upperbound}

In this section we establish that the exponent in the bound from Theorem \ref{thm.stmt2} is asymptotically tight with respect to $\epsilon$ and $w$.

\begin{theorem}
Fix $w$, take $n$ sufficiently large, and take $\epsilon > 0$ sufficiently small.  Then there exists a constant $c > 0$ such that $\Ex[x]{r_T(x)} \leq e^{c \epsilon^2 w^2}.$
\end{theorem}
\begin{proof}
Suppose that some node $y$ is viral in the initial configuration, say with positive bias.  We claim that, with constant probability, $y$ will have positive spin in the final configuration.  To prove the claim, choose $R$ as in the proof of Theorem \ref{thm.stmt2} in Section \ref{sec:mainthm}.  Then, with constant probability, there do not exist any $\epsilon$-negatively-biased nodes within $N_R(x)$.  Invoking Lemma \ref{lem:unhappy} and Lemma \ref{lem:big.ball} with $x = y$, then Lemma \ref{lem:ball.persists}, we conclude that with high probability node $y$ will be contained in a monochromatic region with positive spin in the final configuration, and hence $\sigma_T(y) > 0$ as claimed.

Let $Z^+$ be the set of initially positively-biased viral nodes in the grid whose spin in the final configuration is also positive.  Let $Z^-$ be the corresponding set of negatively-biased viral nodes.
We can conclude that, for any given node $y$, the probability that $y$ lies in $Z^+$ is at least a constant times the probability that it is viral,
which is $e^{-\Theta(\epsilon^2 w^2)}$ (by Lemma \ref{lem:viral.given.unhappy} and the fact that a given node is $\epsilon$-biased with probability at least $e^{-\Theta(\epsilon^2 w^2)}$).  Moreover, the event that $y$ lies in $Z^+$ is independent of the event that $z$ lies in $Z^+$ if the distance between $y$ and $z$ is greater than $2R$, as these events depended only on what occurred within a neighborhood of radius $R$ from each node.  We can conclude that for some $R' = e^{\Theta(\epsilon^2 w^2)}$, any given neighborhood of radius $R'$ will contain a node from $Z^+$ (or, symmetrically, from $Z^-$) with constant probability.

Now choose any node $x$ and constant $r > 0$.  
If $r_T(x) \geq r$, then it must be that at least one of the four axis-aligned squares with $x$ as a vertex, and side-length $r$, is monochromatic.  This event cannot occur if each square contains both a vertex in $Z^+$ and a vertex in $Z^-$.  We therefore have that $\Ex[x]{r_T(x)}$ is at most a constant times the expectation of the minimal $r$ such that a square with $x$ as a corner contains both a vertex in $Z^+$ and a vertex in $Z^-$.  However, as established above, each neighborhood of radius $R' = e^{\Theta(\epsilon^2 w^2)}$ has a constant probability of containing a node from $Z^+$ or $Z^-$.  Thus, the expected radius of a neighborhood that does not contain nodes from both is $O(R') = e^{O(\epsilon^2 w^2)}$.  
%
We can therefore conclude that $\Ex[x]{r_T(x)} \leq e^{\Theta(\epsilon^2 w^2)}$, as required.
\end{proof}

\end{document}